\documentclass{article}

\usepackage{amssymb}
\usepackage{amsmath}
\usepackage{amsthm}
\usepackage{amsfonts}
\usepackage{amscd}
\usepackage[dvips]{graphicx}

\newtheorem{thm}{Theorem}
\newtheorem{cor}{Corollary}

\usepackage[round]{natbib}
\begin{document}


\begin{center}
 \textbf{\large{\textsc{A Critical Discussion About The Methodology Of Quantum Theory}}}

\vspace*{10mm}

\textbf{M. Ferrero}

\vspace*{1mm}

Dpto. F\'{\i}sica, Universidad de Oviedo\\
33007 Oviedo, Spain\\
maferrero@uniovi.es

\vspace*{2mm}

\textbf{D. Salgado }

\vspace*{1mm}

D.G. Metodolog\'{\i}a y Tecnolog\'{\i}as de la Informaci\'{o}n y las Comunicaciones\\
Instituto Nacional de Estad\'{\i}stica, 28071 Madrid. Spain.\\
david.salgado.fernandez@ine.es

\vspace*{2mm}

\noindent \textbf{J.L. S\'{a}nchez-G\'{o}mez}

\vspace*{1mm}

Dpto. F\'{\i}sica Te\'{o}rica, Universidad Aut\'{o}noma de Madrid\\
28049 Cantoblanco, Madrid, Spain\\
jl.sanchezgomez@uam.es

\vspace*{10mm}


\vspace*{10mm}

\textbf{Abstract}

\vspace*{3mm}

\begin{minipage}{10cm}
It is argued that the traditional ``realist'' methodology of physics, according to which human concepts, laws and theories can grasp the essence of reality, is incompatible with the most fruitful interpretation of quantum formalism. The proof rests on the violation by quantum mechanics of the foundational principles of that methodology. An alternative methodology, in which the construction of sciences finishes at the level of human experience, as standard quantum theory strongly suggests, is then conjectured.
\end{minipage}

\end{center}

\section{Introduction}

In a recent book, Lee Smolin defends a conception of science that he calls ``realist'', according to which science should inform us how the real world out there, or nature, really is in our \textit{absence} \citep{Smo07a}. The previous sentence does not reflect only a personal point of view. In spite of the radical change that with respect to physical reality quantum theory has produced, this conception is still dominant inside the scientific community and permeates classical physics as well as all the other sciences. This belief is based on different arguments and ideas that Smolin explicitly or implicitly reproduces. The first idea is the introduction of a \textit{set of principles}, which can be derived from ``pure reason'', as \textit{necessary preconditions} to reach that objective. Einstein was the champion between the defenders of this approach. The second idea is that given that the world out there exists independently of our existence, as one of those principles affirms, the individual persons, or the subject, are interchangeable and eventually can be \textit{segregated}. The third idea is that this world can be known with certainty and, as result of our free inquires, we should get a perfect representation of how nature is. The fourth idea is that nature is unified, there is only one world. Therefore, we should be able to find a final and complete theory, a theory of everything, including all we know.\\

This program, well known by both scientists and philosophers of science, has given to science the greatest success. The advance of knowledge based on it in physics, biology, chemistry, etc., was unprecedented in human history. Even today, with minimal but necessary adjustments, the belief in a real world out there and the possibility of knowing it with certainty in its minute details is, for many scientists, the motivation to carry out the difficult but enthralling work of making science.\\

However, in the first twenty five years of the XX century, the experimental facts shown that in the microscopic realm the previous approach was too limited and, to a certain extent, unjustified. This surprising development was not the consequence of a \textit{philosophical whim}. It was slowly but indefatigably imposed on the scientific community by the progressive accumulation of facts, that were impossible to accommodate inside the previous program. Experiments like those by Rubens and Kurlbaum (1900) on the black body spectrum; Lennard's photoelectric effect (1904); Rutherford, Geiger and Marsden (1909), on the nuclear atom; Bohr's model (1913) and the Frank and Hertz (1914) verification of the discrete atomic energy levels; Stern and Gerlach experiment on the existence of the spin (1922); Compton experiments (1923); Davisson and Germer (1926) on the interference of electrons; etc., all this experimental evidence opened a necessary and new manner of looking at nature, which questioned the same basis of the above realist approach. It was then conjectured that in physics there were two different domains of experience, the micro and the macroscopic; that the \textit{a priori} principles did not work in the microscopic domain; that the subject could not be altogether eliminated and that, if there is a world out there, we could not reach and know its essences. The message seemed to be that, at the microscopic level, physics did not speak anymore about \textit{how the world really is in our absence}, but how the world is in a concrete historical period of the human development. That is, how the world is \textit{in our presence}.\\

This essay does not endorse or develop any previous philosophical work. Rather, it seeks to establish a new approach to the old debate between these two alternatives and, for this reason, it neither does refer to recent philosophical literature, nor engages in current discussions on the philosophy of quantum physics.\\

The paper is organized with the following structure. First, we briefly outline the traditional approach to physics by stating explicitly the \textit{a priori} principles considered necessary preconditions for physical thinking. Part 3 is dedicated to analyze the pros and cons of this traditional methodology. In part four, we show how the standard interpretation of non-relativistic quantum mechanical formalism functions in practice. We illustrate its functioning by using an example taken from quantum information theory, namely entanglement swapping. By analyzing swapping, we will prove in part five that quantum formalism does violate the principles previously introduced. Part 6 is a brief remainder of how in the past, the principles which were firmly established manifested a strong resistance to be eliminated. Part 7 introduces an ``operative'' methodology based not on abstract philosophical analysis, but on the practical use of the theory shown in the previous discussion. This methodology includes an idea of truth without correspondence.\\

\section{Einstein's Critical Realism As The ``Traditional'' Methodology Of Physics}

It could reasonably be maintained that in the \textit{scientific practice} and as far as the last objectives of the scientific inquiry were concerned, and until approximately 86 years ago, there were only ``one methodology of science''\footnote{This is a coarse mode of presenting the problem. The stated methodology is a \emph{general synthesis} of different methodologies of science that have been profusely discussed by philosophers in the last two centuries. We are conscious of this defect, but in this schematic characterization, we adopt the point of view of the majority of the experimental physicists' community.}. This traditional methodology, that even today constitutes the epistemological backbone of many research programs, could be schematically characterized in the following way. It begins by postulating that there is out there a material world, which justifies why we all see the same things \citep{Gar83a}. Then, it adds that men are a very accidental and later descendent of primates. For this very reason, reality is there before us and it cannot depend on our existence. And it concludes by stating that the objective of physics, sciences in general, must be to inform us about how this material world really is in our absence, \textit{as if we were not here}. Concretely, physics has to provide us explanations of why things happen; that is, an understanding and even a spatiotemporal image, a representation, of what is going on in its minute details. More arrogant and ambitious programs add that nature is only one, and, therefore, we should be able to find only one final theory from which everything could be deduced, perhaps men included \citep{Wei93a}. That means that if we were clever enough, and we had sufficient computational power, we could calculate how ``things'' would behave as a consequence of the properties of its last constituents, today quarks and leptons (\textit{reductionism})\footnote{We accept the reductionist program as a start point, but not the arrogant part of the program. Leaving apart knickknacks like life, language, consciousness, etc, impossible to reduce and hence \emph{emergent}, the situation in physics is that the standard model of particles has about 20 free constants and the cosmological one about 15, including dark energy and matter. Each of these constants represents something that we ignore. The state of the art today is that nobody knows, and no theory explains, why they have the values they have. In fact, going beyond the Standard Model to the, at present, deepest formulation of string theory, the so-called String/M theory, there is no question as to whether such constants are not actually ``fundamental'', in the sense that their values could not be determined from first principles, similarly as the Kepler radii in the solar system cannot be calculated in Newton´s (or Einstein´s) theory of gravitation. They play the role of initial conditions lying outside the fundamental laws. }. Last, by not least, it is supposed that the subjects are free to choose the experiments they want to carry out.\\

This approach to physics contains a set of principles that being \textit{a priori}, or deducible from pure reason, to use Kant's terminology, are considered as necessary preconditions to make physics. As it is well known and Einstein liked to insist, without any of them, physical thinking in the familiar sense would not be possible \citep{Ein71a}. These principles are: the principle of \emph{realism} (already introduced above); \emph{causality}; the \emph{continuous time flow}; \emph{individuation}; and \emph{locality}.\\

It is not easy to know when and how these principles were first intuited. The most plausible hypothesis is that they had survival advantages for humankind. No matter how it was, those principles crystallized as a final consequence of the XVI century scientific revolution and its influence remain today with us. Although well known, we will briefly introduce them for completion\footnote{Both, classical physics and relativity theory satisfy all of them. Time, whatever it would be, has the same role in classical and quantum physics, so we will leave it out of our discussion.}.\\

As far as the \emph{principle of realism} is concerned, it is conceptually useful to make a distinction between its meaning in the ``broad sense'' and in the ``restricted sense''. This distinction is irrelevant in classical physics, but, as we will see, is pertinent in quantum physics that is realist in the broad sense, but not in the restricted sense.\\

In the ``broad sense'', this principle postulates the existence of a material reality out there and independent of us. In the ``restricted sense'' it affirms that the world is composed by separable (see below) objects that have properties with well defined values. A measurement of the property, correctly carried out with a suitably calibrated device, will reveal the pre-existing value of the property. In the field of research called Foundations of Quantum Physics \citep{Fre04a}, these properties with well defined values are called, for historical reasons in which we cannot enter now, EPR elements of reality \citep{EinPodRos35a}.\\

The \emph{individuation} principle states that things like particles, fields, etc., that have a real existence independent of the perceiving subject, are arranged in a space-time continuum.  At a certain time, and provided these objects are far apart in the space, they may claim an \textit{independent existence} of one another. This same idea is sometimes stated by saying that the objects in question are separable.\\

The \emph{locality} principle is interwoven with the previous ones. It postulates that if we have two objects A and B spatially separated, in the relativity theory sense, any external influence over A has no direct influence on B. It could be also stated by saying that actions-at-a-distance are not allowed in physics.\\

Once more, it is conceptually useful to distinguish between locality in the broad and in the restricted sense. In the broad sense, is the principle that we have just stated above. In the restricted sense, it declares that \textit{even if these direct influences exist}, they cannot be utilized to send information with superluminal velocity. In this form, it is sometimes called the \emph{non-signaling} condition and, as it is very well known, no conflict between it and relativity theory arises: even if the spatially separated particle B ``\textit{instantaneously feels something}'', this subtle influence cannot be used to send messages. A theory could be local in the restricted sense and not in the broad sense, as it is again the case of quantum theory.\\

Although separability (individuation) and locality are closely related (they coincide in classical physics), they are in general different principles. There exist states in quantum theory that are local but no-separable, as for example some Werner states \citep{Wer89a}.\\

It is usually understood that there is \emph{causality} if every event has a set of antecedent circumstances from which the event follows according to a rule. This statement has two different parts, namely: the initial conditions, to be measured, and the physical law, given by a differential equation. Using this definition, as these equations have a unique solution, the same initial conditions, the same causes, will produce always the same effects\footnote{Bohr, however, used to relate this principle to the conservation principles \citep{Boh58a,Boh63a}.}.  In classical physics, the rule will be given by the equations of motion of Newton, Lagrange or Hamilton, and in non relativistic quantum physics, by Schrödinger equation. This implies, as Laplace noted, that if we knew the initial conditions (position and momentum) of all the particles in the universe, we would know the future for all time. As this definition uses a time sequence and we have left time out of our discussion, in what follows we will understand causality as the closely related principle that declares that \textit{the effects are uniquely and completely determined by its causes}. In physics, this view of nature is called \emph{determinism\footnote{We are forgetting now the chaotic cases. Chaos occurs when a system depends in a sensitive way on its previous state. We will leave this problem outside of our discussion. Nonetheless, it should be pointed out that chaos implies indeterminism only in a practical sense, but not in a strictly fundamental one, and this makes it different from quantum indeterminism, which is really fundamental.}}.\\

A few comments are now in order. The first one is that, although we had said that those principles constitute the epistemological backbone of sciences, classical physics in particular, Newtonian gravitation was non-local in both senses, contradicting our arguments. Yet, it is also true that since its introduction, Newton perceived it as to be simply unintelligible, as a provisional expedient to be eliminated from physics. He even thought about different possibilities to avoid the actions at a distance. In a letter to Boyle and in query 21 of the Optics he speculated, for example, about the existence of a substance, called ether, that would penetrated all matter as a medium for the gravitational interactions.\\

The second has to do with an apparent contradiction related to determinism. Both, Thermodynamics and Statistical Mechanics are at variance with it. However, it is generally accepted that this is not a question of principle, but a practical one.\\

The last comment refers to the old argument that in a total deterministic, Laplacian world there is no way for free will. We would have only the \textit{illusion of being free} to choose. This might be a shared but incorrect conclusion. We cannot enter into this discussion here but Squires, for example, maintains that to have free will, determinism would be a necessary condition \citep{Squ90a}.

\section{Pros And Cons Of The Traditional Methodology Of Science}

The traditional manner of doing science outlined above as a general synthesis is based on the validity of the four principles stated. Our thesis is that only if these four principles, in the restricted and broad sense in its case, work, then it make sense to freely study how nature would be in our absence. This is the mechanism that allows us to think that there is a world out there compound by a plurality of individual objects, which have an independent existence, whose properties have pre-existing values that will be revealed by our operations. From this approach, the external world has, thus, an \textit{immutable legality} to be discovered and reflected in laws. In the old days, these laws were conferred by God. Today it is believed that they are ``inherent'' in the Big Bang, from the beginning of time, and we, humans, are \textit{detached} passive observers that have nothing to do with this legality. Being these laws established from the beginning and forever, our role as outside observers could only be to discover and to reflect them in our theories. If we discover bad laws, as in the phlogiston case, we are doing spurious science. If we discover the good ones, we are doing genuine science. From this point of view, the scientific method consist in gradually uncover the real essence of the phenomena, the objective reality, the ``thing in itself''. This is the most plausible meaning of ``in our \textit{absence}''.\\

As we have already noted, the success of this method until the beginning of the XX century is indubitable. Classical physics worked it out to perfection. For the long period between the scientific revolution and the XX century, we were able to remove our interventions into nature completely and to explain the phenomena as if we were seeing ``the truth'', as if we were detached observers of a structured and finished reality that is just out there in front of us.\\

This same success is the reason why this research program is so attractive that still today the great majority of the scientific community firmly believes that his task is to find out how nature really is. As we have already mentioned in the introduction, it explains with simplicity why we all see the same things. The idea of truth as a fair reflection, (alternatively, as a correspondence with), of what really is. The partial utility of superseded theories; or even why we have little control over which theories seem to be ``temporally'' correct. Yet, it is the origin of the true impulse behind our individual research efforts. However, it also has weaknesses. Let us consider some of them\footnote{We do not think we understand these epistemic-ontological problems better than philosophers do. Many philosophers have been engaged in this type of reflection throughout their entire life. We assume, therefore, that some of them may find the discussion that follows (and the one in pages \pageref{p20} to \pageref{p23}) hopelessly naive.}.\\

The first one is the following. The mentioned success has also given us the greatest confidence that at the macroscopic level, objects are what they look like. But are they really? If that were the case it would be difficult to understand why Newton was unable to discover the electromagnetic fields surrounding him; or why a skilled man like Galileo didn't discover the black holes; or W. Thomson, the dark energy; etc. How can we understand hundreds of facts like those if we begin by thinking that we are detached observers of a structured physical reality that lies just in front of us finished in its minute details? Why are we not able to grasp it as it is, to ``discover'' it in a definite way, once and forever? A clear explanation of the permanent imperfection of our grasping and of the \textit{elusiveness of the ``essences}'' will be much welcome. Note that this sounds also as \textit{a great defeat}: for at least the last 200 years generation after generation of young and clever researches have been unable to see what lies just in front of their eyes.\\

The second is the following. Imagine that we go back in time. We do this kind of exercise to draw conclusions in cosmology, or in the theory of evolution, etc. If we go back 110 years, then quantum and relativity theories disappear. 110 years more, and statistical mechanics, electromagnetism, etc., also disappear. Another 110 years back and even Newtonian classical mechanics disappear. If we follow with this exercise and go even back, not only laws and theories disappear. It is doubtful that 3.500 years ago we had \textit{conscious minds} \citep{Sne82a,Jay00a}; that 10.000 years ago we had \textit{writing}; and it seems that if we go back another thousands years more, even \textit{language disappears}. Is it sound to think about the existence of laws and theories without conscious minds, writing and language? Is it an unjustified belief? Or, is only a \textit{metaphor}?\\

The third one is the following. Let us accept for a moment that ``natural laws'' are ``\textit{laws of nature}''. In that case, they can be only a ``partial set of laws''. It is an undeniable fact that the more important advances in sciences have arisen not through the study of phenomena as they occur in nature, but to the study of phenomena produced by technological means. This is, for example, how historically quantum physics was build up. Its principles were introduced thanks to the previous technological development, and therefore, in our \textit{presence}. They changed completely our conception of how the world is and, nowadays, quantum principles are utilized to conjecture how the world could be\textbf{ }at the beginning, that is, in our absence. Is not the ``traditional'' realist methodology surreptitiously reintroducing the subject\textbf{ }at this point?\\

The fourth one is the \textit{non-trivial change of laws and theories}, implicit in the previous comments. As this has been extensively studied by philosophers of science, we will give only the example of Kepler's laws to illustrate what we mean. The non-trivial changes from Newton dynamics to relativity theory or from classical mechanics to quantum mechanics are also well known examples to argue with. The question is: do Kepler's laws uncover the real essence of the phenomena? From the point of view of the realist methodology, the answer should be yes. However, we know for sure that they work only for short times and in the two-body approximation; that the ellipses are given in a plane, while the heavenly bodies move in a three-dimensional space, etc. To make a long story short: History of physics shows that with the unique exception of current laws and theories, all previous theories have been proved to be limited in one way or another and have become surpassed by the new knowledge and the new order introduced by us. It is what E. Wigner once called the ephemeral nature of physical theories (Wigner, 1967). These changes of laws and theories have also been accompanied by the corresponding change in our descriptions, our images, our explanations and the ``truth'' of what really was going on.\\

These non trivial changes of theories are difficult to swallow. If the laws really were ``laws of nature'' describing the essence of things: why are they so \textit{ephemeral, poor and changeable}? It seems to us that implicit in this methodology is the unexplained idea that in any concrete moment of the human development there are ``parts of material reality'' that are out of our reach, and that we progressively increase our knowledge incorporating new parts. Why it is so? Where do these parts come from?\\

The fifth and more serious difficulty for us to accept this point of view is standard quantum theory. The methodology we are now trying to weight up is implicitly and explicitly based on the four principles stated in part 2: realism in the broad and restricted senses; sufficient reason; individuation; and locality in the broad and restricted senses. Einstein was right in this respect: if we want to find out how nature really is in our absence, then those principles must be satisfied. The problem is that, as we will demonstrate in the next paragraph, they are incompatible with the most productive interpretation of our best theory ever. Hence, the methodology based on them cannot be a satisfactory methodology.\\

\section{A Succint Version Of Entanglement Swapping}

\textit{Entanglement swapping} (ES in what follows) is a phenomenon in which the entanglement between two particles is \textit{teleported} to other distant particles by performing suitable joint measurements and broadcasting the results as classical information allowing so a distant observer to make the appropriate selection. The protocol were proposed in 1993 \citep{ZukZeiHorEke93a} and experimentally verified in 2002 \citep{JenWeiPanZei02a}, using the standard interpretation of quantum formalism (see, for example, \citep{CohDiuLal77a}).\\

Imagine that two sources A and B, separated at a cosmological distance, prepare pairs of particles represented by a singlet state: ${\left| \Psi ^{-}  \right\rangle} _{{\rm i},{\rm j}} {\rm \; =\; \; }\frac{1}{\sqrt{2} } {\rm \; }\left[{\left| 01 \right\rangle} _{{\rm i},{\rm j}} {\rm \; -\; }{\left| 10 \right\rangle} _{{\rm i,j}} \right]$, and that the source A emits particles 1 and 2, and the source B does the same at the same time with particles 3 and 4. The sources also emit one pair per minute. The figure \ref{Fig} shows schematically the hypothetical experiment.

\begin{figure}[h!]
\includegraphics*[width=4.55in, height=2.57in]{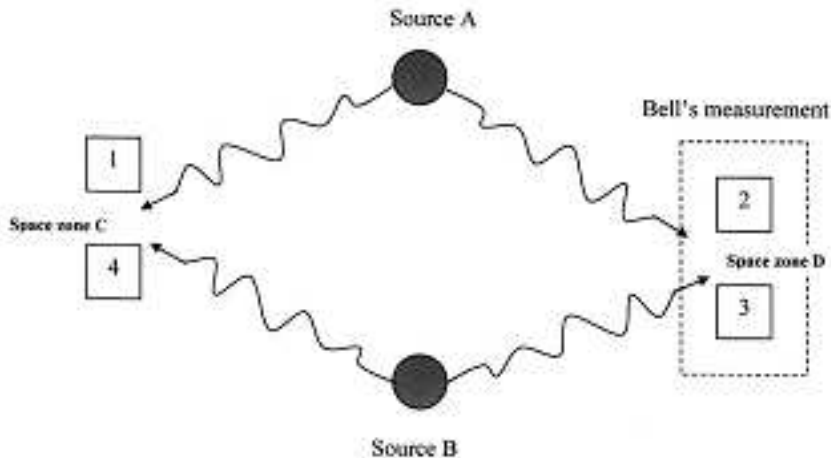}
\caption{\label{Fig}Schematic picture of entanglement swapping.}
\end{figure}

\noindent As the description of the experiment made manifest, particles 1 and 4 have not interacted in the past and therefore \textit{in principle} they have nothing to do with each other. However, 1 is strongly anticorrelated with 2, and similarly 3 with 4.\\

\noindent The joint state of the four particles is given by:
\begin{eqnarray}
{\left| \Psi {\rm \; } \right\rangle} _{1,2,3,4} &=&\frac{1}{2} {\rm \; }\left[({\left| 01 \right\rangle} {\rm \; -\; }{\left| 10 \right\rangle} )_{1,2} \otimes {\rm \; (}{\left| 01 \right\rangle} {\rm \; -\; }{\left| 10 \right\rangle} )_{3,4} \right]   \nonumber\\
&=& {\rm \; }\frac{1}{2} {\rm \; }\left[{\left| 0101 \right\rangle} {\rm \; -\; }{\left| 0110 \right\rangle} {\rm \; -\; }{\left| 1001 \right\rangle} {\rm \; +\; }{\left| 1010 \right\rangle} \right]_{1,2,3,4}                                  (1).\label{JointState}
\end{eqnarray}
This joint state ${\left| \psi  \right\rangle} _{1,2,3,4} $contains complete information about all possible results of measurements carried out upon any part of the four system particles. By this we meant that if you determine the state of any part of the system, the other part would be left in the \textit{relative state} corresponding to the one you have found, as expressed by \eqref{JointState}. For example, if you measure only the spin of the particle 1 in the direction OZ and you get the result${\left| 0 \right\rangle} $, then the particles 2, 3 and 4 will be found, in the corresponding joint measurement, in one of the states ${\left| 101 \right\rangle} _{2,3,4} $or ${\left| 110 \right\rangle} _{2,3,4} $with equal probability. You will never find ${\left| 111 \right\rangle} _{2,3,4} $or ${\left| 000 \right\rangle} _{2,3,4} $, let us say.\\

If now you make instead the measurement $\hat{\sigma }_{z2} \otimes \hat{\sigma }_{z3} $ upon the pair 2 and 3 in zone D of space, the particles 1 and 4, which might be \textit{spatially separated}, ``instantaneously'' become \textit{correlated at-a-distance in the direction OZ.} For example, imagine that you get in your measurement the state ${\left| 00 \right\rangle} _{2,3} $. If someone carries out in that very instant and in zone C of the space the measurement $\hat{\sigma }_{z,1} \otimes \hat{\sigma }_{z,4} $, she will find with certainty that both spins 1 and 4 are down, that is, she will find the relative state\footnote{This is evident if you write down the state \eqref{JointState} in the  following form: \begin{equation}\label{JointState2}{\left| \Psi {\rm \; } \right\rangle} _{1,2,3,4} {\rm =\; }\frac{1}{2} \left[{\left| 01 \right\rangle} _{1,4} \otimes {\left| 10 \right\rangle} _{2,3} -{\left| 00 \right\rangle} _{1,4} \otimes {\left| 11 \right\rangle} _{2,3} -{\left| 11 \right\rangle} _{1,4} \otimes {\left| 00 \right\rangle} _{2,3} {\rm +}{\left| 10 \right\rangle} _{1,4} \otimes {\left| 01 \right\rangle} _{2,3} \right].\end{equation} } ${\left| 11 \right\rangle} _{1,4} $. This \textit{``formal''} prediction can be tested. To this end, observer in zone D must send to observer in zone C information about her results. Otherwise, the observer in zone C, will have no idea neither about the direction, nor upon which concrete subsystems she has to measure.\\

If we carry out a different measurement upon particles 2 and 3, for example a Bell measurement, particles 1 and 4 become now \textit{maximally correlated in any direction\footnote{Again, this is trivial if you write \eqref{JointState2} above in the following form: \begin{equation}\label{JointState3}{\left| \Psi  \right\rangle} _{1,2,3,4} =\frac{1}{2} \left[{\left| \Psi ^{+}  \right\rangle} _{1,4} \otimes {\left| \Psi ^{+}  \right\rangle} _{2,3} -{\rm \; }{\left| \Psi ^{{\rm -}}  \right\rangle} _{1,4} \otimes {\rm \; }{\left| \Psi ^{{\rm -}}  \right\rangle} _{2,3} -{\rm \; }{\left| \Phi ^{+}  \right\rangle} _{1,4} \otimes {\left| \Phi ^{+}  \right\rangle} _{2,3} +{\left| \Phi ^{-}  \right\rangle} _{1,4} \otimes {\left| \Phi ^{-}  \right\rangle} _{2,3} \right]\end{equation} \noindent where  ${\left| \Psi ^{\pm }  \right\rangle}$ and ${\left| \Phi ^{\pm }  \right\rangle}$  are the four Bell states:  ${\left| \Psi ^{\pm }  \right\rangle} = {\rm \; }\frac{1}{\sqrt{2} } {\rm \; }\left[{\left| 01 \right\rangle} {\rm \; }\pm {\rm \; }{\left| 10 \right\rangle} \right]$ and ${\rm \; }{\left| \Phi ^{\pm }  \right\rangle} {\rm \; =\; }\frac{1}{\sqrt{2} } {\rm \; }\left[{\left| 00 \right\rangle} {\rm \; }\pm {\rm \; }{\left| 11 \right\rangle} \right].$}}. The correctness of this prediction can be verified experimentally by testing the CHSH inequality \citep{ClaHorShiHol69a}. If we use pairs (1+4) that correspond to particles 2 and 3 not measured, the inequality would be satisfied. However, if we broadcast the results and, using this information, someone selects in zone C of the space the sub ensemble of pairs (1+4) that corresponds to pairs 2 and 3 that have been submitted to a Bell measurement with the result ${\left| \Psi ^{-}  \right\rangle} _{2,3} {\rm \; }$, for example, then the inequality would be violated. Pairs (1+4) have become now \textit{entangled at a distance}. This post-selection procedure is the one used to test the prediction. We should not underestimate the importance of this protocol. ES is predicted to be a basic element in quantum repeaters in future information technologies.\\

What is new and surprising in the previous description is that, before this protocol was introduced, it was thought that entanglement needed the interaction between two or more particles. Nonetheless, the formalism predicts that particles can be entangled without \textit{direct} interaction. As we have explained, all the information is contained in the whole, the state ${\left| \psi  \right\rangle} _{1,2,3,4} $. This information \textit{combined }with the new knowledge obtained by measuring in zone D, determines the relative state of any sub ensemble in zone C, and their internal correlations. Therefore, we conclude that provided this information is transferred, the\textit{ type of correlation} of each individual subsystem (1+4) is established instantaneously at-a-distance by measuring over the subsystem (2+3) and, hence, that the pure state that describes afterwards each individual pair (1+4) is different depending on the kind of measurement and the result obtained in a spatially separated region of the space.\\

A question now arises, namely: what is the ``cause'' of entanglement? The following inferences could be made in this respect. If there were no selection, then the final state of the ensemble made up by all the particles (1+4) corresponding to pairs (2+3) that have been submitted to a Bell measurement would be:

\begin{equation}\label{JointState4}
\rho _{{\rm 1,4}} {\rm \; }={\rm \; tr}_{{\rm 2,3}} {\rm \; }\rho _{{\rm 1,2,3,4}} {\rm \; }={\rm \; }\frac{{\rm 1}}{{\rm 4}} {\rm \; }\hat{{\rm I}},
\end{equation}

\noindent where $\hat{{\rm I}}$ is the unity matrix in Hilbert space of dimension 4. This state associated to particles (1+4) is \textit{separable}. This implies that we are not going to have a violation of any Bell inequality \citep{Bel87a} using pairs in this state. The fact that it be possible to choose a sub ensemble of pairs (1+4) that violates Bell's inequality, manifests that pairs that have never met before are indeed entangled \textit{after} carrying out measurements at a spatially separated region. The broadcasting of the results (classical information) and the posterior selection of pairs (1+4) based on them are therefore \textit{necessary conditions\footnote{This information is, in our opinion, \emph{the relevant element} of the famous sentence that Bohr wrote in 1935 answering the EPR paper: ``But even at this stage, there is essentially the question of \emph{an influence on the very conditions which define the possible types of predictions regarding the future behaviour of the system}'' (emphasis by Bohr, in \citet{Ein70a}, p.\ 234).}} for the swapping protocol to be realized, but they are not \textit{sufficient}. Let us argue this point more carefully. The issue to address is this: is it the transfer of classical information and the posterior selection the ``cause'' of the entanglement we are looking for, or can we make the (im)plausible hypothesis that the pairs (1+4) were entangled ``instantaneously'' by the successive Bell measurements? We know for sure that if no Bell measurements were carried out upon particles 2 and 3, no entangled pairs (1+4) at all could be selected. We also know that if the Bell measurements are carried out but no information is broadcast, and hence no selection is done or is done by chance, no violation of CHSH inequality would be produced. Therefore, we conclude that the broadcasting of information plus post selection are necessary conditions to entangle the pairs but that it is only \textit{their conjunction} with the Bell measurement what makes all three sufficient for ES. Can we go beyond and infer that somehow the entanglement between particles (1+4) is produced by the \textit{physical interaction} between particles 2 and 3 brought about by the measurement upon them? We do not have a clear answer but, if we look at the expression \eqref{JointState3}, we see that the possibility that they were instantaneously entangled by those measurements is included in the initial state ${\left| \psi  \right\rangle} _{1,2,3,4} $. This state predicts that if particles 2 and 3 interact in the appropriate way, the effect of this interaction is somehow \textit{transferred }to the subsystem (1+4). This inference is reinforced by the fact that if the pairs were really entangled one by one with each Bell measurement, the density matrix $\hat{\rho }_{1,4} $ representing the state of all that pairs, would be an equal mixture of the four Bell sates, as expression \eqref{JointState3} predicts. Now if you make this simple calculation you will get expression \eqref{JointState4}, supporting our previous argument. If this inference sounds peculiar, as may well happen, the conceptual ``interference'' would come, once again, from common sense. We are accustomed to think that, for example, money plus money gives more money, but here we see that this ``true'' is not favored by quantum formalism: entanglement plus entanglement might give a separable state, that is, no entanglement at all. Another quantum surprise!\\

We reach therefore, the following conclusions from our analysis\footnote{The account we have given should be sufficient to understand why this protocol cannot be exploited to send signals at superluminal velocity.},

\begin{itemize}
\item[1${}^{st}$]: that the relative state of particles (1+4), and so the correlations to be uncovered between them, depend on the measurement and \textit{choice} made at a spatially separated zone of the space. And,\label{p13}

\item[2${}^{nd}$]: that the measurement does not ``reveal'' a particular set of \textit{previous} correlations. The correlations between particles 1 and 4 do not have their origin in an interaction in the past. They must, accordingly, be ``created\textbf{''} in the measurement. However, if they are created in the measurement, then one of two things: either there is not ``enough reality'' to account for what occurs \citep{Lap07a} or material reality is partially undefined. In any of these two cases, it is obvious that we are not ``getting closer'' or ``grasping more fully'' or establishing an impossible ``one-to-one correspondence'' with something that still does not exist.
\end{itemize}
\section{The Breaking Of Classical Principles (and Einstein's Critical Realism)}

We briefly explain now how the classical principles introduced in part 2 are violated in this concrete example. \emph{Realism in the restricted sense} is violated because the individual particles 1, 2, 3 and 4, do not have spins with a well defined value in any direction, a pre-existing value that the measurement would reveal. The sense of the spin in any direction is ``defined'' (created) in the same measurement, so EPR elements of reality do not pre-exist \citep{Mer93a}. As just explained above, the same can be predicated about the correlations: they are \textit{not new} EPR elements of reality either.\\

\emph{Determinism} is violated because the effect is not completely and uniquely determined by its causes. The technician making Bell measurements in zone D and knowing the joint pure state \eqref{JointState3} will get, \textit{for exactly the same initial conditions}, the result ${\left| \Psi ^{+}  \right\rangle} _{{\rm 2,3}} $ in 25\% of the cases, the result ${\left| \Psi ^{-}  \right\rangle} _{2,3} $ in 25\% of the cases, and so on. There is nothing she can do to improve or modify the statistic. Quantum theory is an essentially probabilistic theory, whatever ``essential'' might signify here. The individual results occur without deterministic law. This is crucial, because, as we have already emphasized, it allows the compatibility with relativity theory: signaling is impossible using quantum formalism.\\

The \emph{locality principle in the broad sense} is violated because the state of the particles 1 and 4 can be prepared at a distance by a measurement upon particles 2 and 3 in the region D.\\

Finally, the \emph{individuation principle} is violated because in this phenomenon \textit{is not possible to separate }the state of particles (1+4) from the\textit{ interaction }of particles (2+3) with the \textit{measuring device.} That is, at that precise moment, particles 1 and 4 ``feel something'' that entangles them. As previously discuss, we must confess that we do not have a concrete proposal to understand what this ``telepathic coupling'', to use Einstein words \citep{Jam74a}, p. 494, could be.\\

We have utilized entanglement swapping to demonstrate the violation of these four principles. This is only one example. As a matter of fact, any scholar working in the field knows that these experimental violations are contained in the principles postulated by quantum theory. In our opinion, the so-called ``weirdness'' of quantum theory is related to this violation of the principles of realism in the restricted sense, determinism, locality in the broad sense and individuation.\\

The distinction we have made between broad and restricted sense in two of the principles, has been introduced to help to eliminate some confusion about the nature of the quantum principles. For many researchers, quantum theory is both realist and local, while for many others it is, neither realist, nor local. There is really no contradiction at all. They might be even saying the same thing. If we take realism in the restricted sense, quantum theory is not realist, but if we take it in the broad sense it is indeed realist. Something similar could be predicated about locality. If we take it in the broad sense, quantum theory is non-local. However, if we take it in the restricted sense, it is local. Another form of expressing the same thing would be to say that quantum theory is a realist theory, but the EPR elements of reality do not exist, and that it is non-local but this non-locality does not allow sending information at superluminal velocity. This summarizes, and we hope, clarifies, the contradictory opinions that can be read in the foundations of quantum physics field.\\

There are some variants of ``realist interpretations'' and alternative formalism, like Bohm's theory \citep{Boh52a}; dynamical collapse theories \citep{GhiRimWeb86a}; transactional interpretation of quantum mechanics \citep{Cra86a}; etc., that uphold some of the previous principles, but not all of them. The more they uphold, the more expensive is the price they have to pay. It is not the objective of this paper to discuss these alternatives and to analyze, for example, how they understand entanglement swapping. We see them as \textit{defensive strategies} that philosophically look to the past, not to the future. In fact, they have been always \textit{behind} the developments of standard quantum theory, like for example in the case of quantum information, and in the last four or five decades they have been unable to produce anything fundamentally new. This is the reason why, rather than looking back and try to defend how the world is in our absence, our objective will be to disclose \textit{how the world could be if the standard interpretation of quantum formalism is considered correct}.\\

\section{The Resistance Of Established Principles To Be Relegated}

Limiting our arguments to physics, the ``most fundamental'' of all sciences in the uncritical reductionist program, it is an historical fact that in the past, whenever some principles showed their limitation to encompass new empirical evidence and new principles were introduced to do it, a strong resistance to abandon the old principles and to embrace the new ones emerged. This \textit{conservative attitude} should be applied to the new principles, to allow them to show their potentialities, but it is usually doomed to failure with the old ones.\\

The arguments used today to reject standard quantum formalism can be summarized by saying that, to the detractors, the new quantum principles are regarded as doubtful, confusing, obscure and even absurd\footnote{ This is again a crude exposition. A detailed analysis would take us far away from our main objectives.}. Even if the theory is coherent and in perfect agreement with all the observed facts, as it is the case, its ``true meaning'' is disputed. Besides, part of the empirical evidence that supports it, is carefully scrutinized to find loopholes. Good critical work is always welcome \citep{San09a}. However, the main objective of some of these maneuvers seems to maintain safe the old established principles and the common sense related to them. In these last cases, the acceptance of the new theory does not depend on its merits, but upon whether it is compatible with the established principles or not. If it is not, as it is the case with quantum theory, then it cannot be the ``truth'', and must be completed, revised or rejected, as in the alternative proposals quoted above. This is nothing new. On the contrary, it seems to have been the normal rule in the history of physics. Let us remember briefly the well known examples of Copernicus and Newton.\\

Despite its agreement with the observed phenomena, the Copernicus system was strongly criticized on the basis that its principles were apparently absurd. This was the shared opinion not only inside the Catholic Church, but also by scholars like Francis Bacon, one of the pioneers of the empirical science. For Bacon, Copernicus' principles were ``the speculations of one who cares not what fictions he introduces into nature, provided his calculations answer'' \citep{Bac53a}. One century later, Newton's theory got a similar criticism. Not only from Berkeley. Leibniz, a practicing scientist who made important contributions to mathematics and physics, considered both, the inertia principle and the gravitation law as absurd and even false. In both cases, Copernicus and Newton, the real reason behind the criticism and the rejection was \textit{uniquely} the contradiction between the new principles introduced by their respective theories and the Aristotelian well established principles.\\

If we were interested in learning the lesson from the past, it would be a good exercise to have a look at what happened afterwards. Two things are interesting in this respect. In the following decades, the range of physical phenomena that were in perfect agreement with the new principles increased dramatically. On the other hand, those new principles allowed new theoretical derivations that gave a more profound understanding of the theory. These two facts changed slowly but inexorably the attitude of the next generations, to the paradoxical point that, after another century, and until today, the principles introduced by Copernicus and Newton were considered as ``self evident principles''. Let us make only two quotes. In the \textit{Critique of Pure Reason}, Kant declared that objective empirical knowledge could only be achieved if individual laws were formulated according to Newton's principles. And half a century later, Helmholtz wrote: ``The task of physical science is finally to reduce all phenomena of nature to forces of attraction and repulsion, the intensity of which is dependent only upon the mutual distance of material bodies. Only if this problem is solved, we can be sure that nature is conceivable'' (quoted from \citet{Whi50a}).\\

This brief historical, and to some extent slanted, reminder shows something that we all know: established principles, laws and theories have always been displaced by the new ones, and in many cases the new ideas were strongly criticized because they violated the previous ones. It is a fact that physical theories have always been changed when our interventions into nature establish new relations that are neither predicted, nor embraced by the current theory in any convenient way. If we consider that there are philosophical principles that can be derived from pure reason, it is evident that they can never be modified. But this is not what has happened.\label{p17}\\

We should not underestimate this lesson from the past. The argument by Einstein that the principles of realism, individuation, determinism and locality are necessary ``preconditions'' to postulate laws which can be checked empirically \citep{Ein71a}, is similar to the arguments made by Berkeley (and by Kant). There is only a change in the principles. When people claim that quantum theory is a valuable description of observed facts that does not give us neither a ``real understanding'', not a description of what is ``really'' going on, an ontology of physical reality, they are reproducing Bacon attacks. The real meaning of all this criticism is that the principles of quantum physics, the formalism compacted in four or five postulates, are in contradiction with those ``evident and well established'' previous principles.\\

\section{A Critical Discussion About The Methodology Of Quantum Physics}

We have seen that quantum theory strongly suggests that the traditional methodology is inadequate at the microscopic level, thus invalidating its role as a general methodology. The alternative is then: either non relativistic quantum theory in its standard version is wrong, or the traditional methodology of science based on those principles is misleading. In what follows we will take the second alternative and we will try to infer how the world could be like if the standard interpretation of quantum theory were true.\\

To this end we will adopt the strategy that the traditional methodology does in fact work at the macroscopic level but not in the microscopic one. We will keep then realism in the broad sense, but we will reject it in the restricted sense. The immediate ontological outcome of this decision is to put aside the EPR elements of reality. However, if the systems do not have preexistent defined values for a concrete property, let us say the spin, then the values must be \textit{created} in the same act of measurement. The consequences are twofold.\\

On the one hand, we cannot maintain the program of how the world really is in our absence. If the kind of correlation in ES or the EPR elements of reality are \textit{created} by measurement, \textit{our presence }becomes essential. The change is then dramatic: we pass from being in a cool and alien universe, in which we are mere descriptive spectators, to a \textit{participative }one \citep{Whe96a}, one in which our free decisions count and therefore the subject cannot be altogether eliminated.\\

On the other hand, the material world does not manifest to us as something finished in its minute details, but as something \textit{partially undefined}. Remember that quantum entities can behave like particles or waves, depending on \textit{our choice}, on how \textit{we} \textit{decide} to observe them. This indefiniteness evokes us Dirac's statement about material reality. According to Dirac, the fundamental laws of nature, as stated by non relativistic quantum theory, do not govern the world as it appears in our mental picture in any direct way, but they control ``a substratum of which we cannot form a mental picture without introducing irrelevancies'' \citep{Dir30a}.\\

\label{p18}

It is convenient at this point to make a distinction between \textit{material reality} and\emph{ }\textit{physical reality}. Material reality, or matter is by definition the undefined Dirac substratum, and physical reality is a \textit{social construction} that we build up based on our current laws and theories. Let us put the stale example of the desk. For a physicist it is almost empty space: atoms in a particular order. However, for people drinking beer in a pub it is a defined piece of wood, a solid object. Rigorously speaking, things are not what they seem to be to the common sense. Sciences define what things are. If theories of science change, what things are, and this is valid from the tiniest parts to the whole universe, change accordingly. We do not know any \textit{permanent} physical reality. The material reality, or the substratum, might be the same the physical reality is not. The reason behind this distinction is that our concepts and the substratum are, presumably, very different things. Concepts are associated with our personal and social experiences. The substratum is not. Therefore, there is no a priori reason to think that the essence of material reality can be reproduced directly in the world of our laws and theories. How could we justify a supposedly isomorphism between two things of such a different nature? We cannot. And if there is not an isomorphism, the ideas of reflection and correspondence become not only obscure, but incomprehensible. This is why we adopt the point of view according to which our laws and theories allow us to introduce order in our experience. But the success in introducing order in our experience does not secure the correspondence between theories and the substratum, as manifested by the non trivial change of theories. The object of physics would not be then to make a mental or mathematical ``image'' reproducing material reality in its minute details, a one-to-one mapping, but \textit{to increase its utility for us by putting order in our experience}.\\

As far as the other principles are concerned, we reject them all. This takes us to the following situation. We have now the principle of realism in the broad sense, plus quantum mechanical formalism. Our task is to sketch a homogeneous methodology of physics based not in a philosophical erudite analysis, but looking carefully at how quantum theory really works, and drawing the pertinent conclusions from there. We have seen how it works in part 4 of this essay. The first thing we want to call the attention to is the fact that in part 4 the formalism was applied only to the isolated system  ${\left| \psi  \right\rangle} _{1,2,3,4} $, but neither to the sources, nor to the devices that were performing the measurements. The preparing and measuring devices are described in terms of the operations that the technicians, that have previously calibrated the apparatus, ``know how'' to do to get the results. This amounts to posit the inapplicability of quantum formalism to some part of the measurement process. To be more concrete, eighty-six years of research have shown that, due to the very existence of the quantum, it is not possible to \textit{control the reaction of the object on the measuring device}, as Bohr anticipated (Bohr, 1958). This fact splits the world in two parts: the quantum system to which the theory applies, and the rest of the world, in particular the measurements apparatus, to which the theory does not apply. This partition of the world into ``systems and observers'' with \textit{an arbitrary line}, was very unsatisfactory for many people, John Bell being one of them. For Bell ``the possibility of a homogeneous account of the world'' \textit{recovering the old principles}, was the chief motivation to study hidden variable theories (Bell, 1987. See also Jammer,1974, p. 253). Lack of space does not allow us to justify properly why we contemplate this understanding of the formalism with a division line as the one with the lesser conceptual difficulties. Loosely speaking it coincides with the Copenhagen interpretation of quantum theory as explained mainly, but not only, by Bohr \citep{Boh58a,Boh63a}, Heisenberg \citep{Hei58a} and Stapp \citep{Sta72a}, and also with the variant introduced by von Neumann, London, Bauer and Wigner. We will justify this choice in a separate paper.\\

Let us go back to our main objective. Our task was to elaborate a homogeneous methodology of modern physics according with realism in the broad sense and the new principles introduced by non relativistic quantum theory, without any other commitments. As discussed above, our approach is that these principles are applied only to the quantum phenomenon, not to the devices and, therefore, a division line between the subject and the object must be introduced. Consequently, the interchangeable subject cannot be altogether removed, conclusion that we had already reached by a different path. This is a strong indication that, to use a metaphor, in the same manner that the land of an island finishes at the mobile water edge, the construction of sciences finishes at the edge of our interventions, not at the edge of the essences, of the thing in itself. Physics is compatible with the postulation of the thing in itself, but it does not need this postulate to make sense. The thing in itself can be understood as a desired referent, but definitely it is not what sciences reach. It is in this sense that there is no science of the essences, or science in our absence: the construction of sciences finishes at the level of human experience. This is for us one of the epistemological lessons to extract from quantum theory.\\

We summarize our arguments about the indispensable role of the subject in making science in a theorem:
\label{p20}
\begin{thm}
The events known by nobody, do not exist (for science).
\end{thm}
\begin{proof}
Obvious.
\end{proof}

\begin{cor}
Unperformed experiments have no results.
\end{cor}

\begin{proof}
See \citet{Per78a}.
\end{proof}

A brief comment about the theorem seems to be now necessary. In cosmology the Universe is observed not as it is at this very moment, but as it was a few years ago, or 4 billions of years ago, and based on this information and general relativity, we make suitable assumptions avoiding contradictions. As the information changes, we modify accordingly \textit{what we can say} about it. This evokes the Berkeley's old question: has the tree fallen in the forest if nobody was there? From our point of view the answer was given some 40 years ago by von Weizsäcker: ``what is observed, certainly exist; about what is not observed we are still free to make suitable assumptions. This freedom is then used to avoid paradoxes'' \citep{Wei71a}.\\

We have gathered in the previous pages, by a close examination of our best theory, enough elements to advance in our purpose to conjecture an alternative methodology. Those elements are:

\begin{itemize}
\item The principle of realism in the broad sense.

\item The principles of quantum mechanics, as used in part 4${}^{th}$.

\item The mutability of laws and theories.

\item The conjecture that physical reality is indeed an objective social construction (explained in page \pageref{p17}).

\item The division line between macroscopic and microscopic levels, introduced in page \pageref{p18}.

\item The fact that observer-participancy matters. If the values of the physical magnitudes (or the correlations) are created in the measurement, as we have shown with ES, then our role and our choices deciding with measurement to take are relevant. (See page \pageref{p13} above).
\end{itemize}

A general methodology of science incorporating those elements could be succinctly conjectured as follows. Since the mists of time, humans have interacted with their environment (with material reality, whatever this would be), first with their hands, sense organs, ideas or beliefs and then with their technological devices. In this process, undifferentiated material reality was decomposed into pieces and composed once again differently. We began to hold pieces of matter, to manipulate them (bones and stones as weapons, for example), to make material operations with our surrounding material reality. It was this manipulation towards survival which, in the becoming of time and over many thousands of years gave rise, from that undifferentiated material reality, to pieces of matter, then to ``unconscious concepts'', and finally to something that we now regard as ``material objects''. Objects are, therefore, seen here as the result of an efficient interaction with the environment done with survival purposes.\\

It was the subsequent manipulation with these material objects what established recursive relations. These relations were used to build up new material objects and new relations, and so on. Undifferentiated material reality gave way to ``unconscious concepts'', then to physical objects ``as postulated entities which simplify our account of the flux of existence'' (Quine 1980), and finally to relations summarized in a certain algorithm able to reproduce all relations of that class. This is what we designate as a \textit{natural law}. From this ``operative mindset'', the scientific laws appear not as \textit{inalterable ``laws of nature}'', but as our own constructions: computational algorithms that allow to condensate and reproduce an enormous variety of relations between objects and between macroscopic devices. It was these everlasting interventions in a malleable nature what shaped our world, introduced order in material reality and \textit{mental evolution} in human beings. The most plausible concatenation seems to have been: unconscious concepts $\to $ objects $\to $ sign language $\to $ spoken language $\to $ writing $\to $ conscious mind $\to $ laws and theories. We do not need to invert this sequence to understand the origin of our theories. Unconscious concepts, objects, language, consciousness, laws and theories appear in this methodology as social constructions, submitted to change and evolution because each of them are products of the previous order that we had introduced in any concrete field of human activity. Science is seen here in continuity with all other human activities: a survival activity resting on preceding knowledge obtained depending on both, the material substratum and the particular state of evolution of our extended mind.\\

The described sequence, with ``unconscious concepts'' at the beginning and conscious minds almost at the end needs an explanation that we succinctly give now. We can introduce it with the following question: what came first, concepts or objects? Can we have ``objects'' before we have concepts? In our opinion, the answer should be an unqualified yes. Let us explain this point with a few examples. The antelope has the concept of ``lion''. It runs away as soon as it sees one close by. The same can be affirmed in the case of the bees. They have the concept of ``flower''. Both, the antelope and the bee, have the ``concepts'', but they are not conscious concepts \textit{in the same manner} as we have them\footnote{We are indebted to Prof. T. Calvo for calling our attention to this point.}. These are the ``unconscious concepts'' introduced above. If our approach is correct, the unconscious concepts predate objects and conscious minds: concepts do not require consciousness. Jaynes' theory, which we follow in this concrete point, contends that it is possible to conceive of humans with all the traits of learning, reason, language, and so on, but not consciousness \citep{Jay00a}.\\

Note that, in the previous scenario and as far as the material reality concerns, the order introduced is not arbitrary. Not everything goes. In fact, it is independent of our individual will. Matter has its own \textit{legality}, revealed by its \textit{resistance} to our operations, by the fact that some of the things we do work and others do not. Those that work show invariants, regularities that we express in laws. These empirical pressures are the sure indication that material reality has legality. Quantum theory reminds us that this legality has manifested itself only through \textit{our interventions}.\\

This approach allows understanding why the traditional methodology worked so well in the past. It was due to the fact that then the level of organization of human experience was \textit{incipient} and the division line between subject and object was unnecessary. In such a situation we could remove our interventions completely. While that happened, the traditional methodology was sufficient and the phenomena were explained as if it were given in our absence. The second epistemological lesson of quantum theory is that actually the process is just the contrary. It is only once we have objects, language, consciousness, laws and theories that we can \textit{give the jump} and \textit{contrive} how the world would be in our absence. But this \textit{re-construction }requires the introduction of a subject that has been not a passive observer, but \textit{an active agent trying to solve old problems with new methods and new problems with creativity and imagination}.\\

Let us finish our proposal addressing a point that, perhaps, some readers have now in mind. If the construction of sciences finishes at the level of human experience, what about the ``truth'' of our theories? In the macroscopic side of the line, the truth is what naïvely seems to be: the correspondence with the object. At this level, traditional methodology is enough. However, at the microscopic level, there is not correspondence at all and that methodology is silent: it has nothing to say.\\
\label{p23}
The concept of scientific \textit{truth without correspondence} related with the general methodology outlined here should satisfy two conditions previously introduced, namely: it must be useful in the microscopic and in the macroscopic realms and it must combine material legality with the fact that the scientific construction finishes at the edge of human experience. Although the thorough discussion of this subject would lead far beyond the limits of a single essay, we will briefly introduce it using a few intuitive examples covering both realms.\\

In this approach, the concept of truth is established by what it is called \textit{confluence}. A confluence is produced when, using different and independent processes, we reach inexorably the same result. Independent means that one of the processes could be overturned without the others being necessarily so. Confluence is then more than mere \textit{regularity}. Regularity is a necessary condition satisfied in every separate independent process. Each of them is repeatable and regular, but this is not sufficient to be a confluence. It requires the convergence of various regularities. For example, the determination of Planck's constant can be done based on many different research programs: the black body radiation; the photoelectric effect; the Franck-Hertz experiment; the X- Ray production; etc. When these determinations are done, it is found that, within experimental errors, the regular results obtained \textit{in each} independent process converge in the same value for \textit{h}. This confluence establishes the consensus inside the physicists' community, ``materializes'' the quantity \textit{h}, and therefore it is said that \textit{h }is a truth of physics. We could argue similarly in the case of the Avogadro constant, determined by more than 30 different and independent methods and many other cases. Note that, although we have taken two examples from physics that are constants, belonging to the micro and macroscopic realms respectively, the confluence process does not refer only to physical constants. Is the \textit{general manner} of establishing the concept of truth. For example, the movement of the Earth's continents relative to each other (Continental drift, Wegener 1912) can be considered today a scientific truth\footnote{Although previously anticipated it required the theory of plate tectonics, established only about 1960.} based on the confluence of independent and regular evidences like: identical fossils of plants and animals that can be found in Africa, South America, India and Australia; the distribution of the same kind of glacial sediments in all these regions; the distribution of equal species that, like the snails, could have not travel through the oceans; the same erosive model in different continents; the fact that the borders of the continents fit together; etc.\\

It is a historical fact in science that whenever the same result has been reached at least by three or four independent means, it has never been overturned. This implies that, as we extend human experience by the process succinctly described above, the number of confluences never decrease lying down the foundations of what it is understood as the \textit{scientific progress}. These confluences besides being permanent are \textit{objective}: is not in our hands to make them converge or not. That is, they cannot be manipulated by the desires of an individual or a community. They are more than intersubjectivity or consensus inside a scientific community. They are the indelible proof that matter can be organized by our activity according to its own resistance and the simplest explanation of why scientists are so successful at reaching consensus. Confluentism in physics seems to be able to do, for example, what R.\ Rorty said that it could not be done: to have an objective criterion of  truth putting aside at the same time the correspondence theory\footnote{``No interesting connection will ever be found between the concept of truth and the concept of justification\dots truth as the aim of inquiry\dots  is either empty or false''  \citep{Ror99a}, p. 37. }. It shows that really there is something epistemological special about the nature of scientific knowledge. It allows maintaining contact with the material reality partially undefined as well as holding that science is a social construction, an idea shared with pragmatists like Rorty and -with due qualifications- with the strong program of the Edinburg School \citep{Bar79a,BarBloHen96a}. What seem to have misled these authors is that this same procedure is used in all our (social) activities, like jurisprudence, psychology, history or medicine, to name but a few. However, the big difference, not considered by these authors in all its relevance, is that, in physics, an appeal to the resistance of matter, material legality, is realized. Nevertheless, this is all we need.\\

\section*{Acknowledgements}

We thank an anonymous referee for her helpful criticism in a previous version. We also acknowledge financial support from the Spanish Ministry of Science and Innovation under projects: MICINN-08-FIS-2008-00288 and MICINN-FIS2011-24885. MF also acknowledges a sabbatical leave from Oviedo University.


\begin{thebibliography}{99}
\providecommand{\natexlab}[1]{#1}
\providecommand{\url}[1]{\texttt{#1}}
\expandafter\ifx\csname urlstyle\endcsname\relax
  \providecommand{\doi}[1]{doi: #1}\else
  \providecommand{\doi}{doi: \begingroup \urlstyle{rm}\Url}\fi


\bibitem[Bacon(1653)]{Bac53a}
Bacon, F. (1653). \textit{Descriptio Globi Intellectualis}, in \textit{The Works of Francis Bacon}, Vol v. page 517. Elibron Classics. (2006).

\bibitem[Barnes(1979)]{Bar79a}
Barnes, B. (1979). \textit{Interests and the Growth of Knowledge}. Rutledge. London.


\bibitem[Barnes \emph{et al.}(1996)]{BarBloHen96a}
Barnes, B., Bloor, D., and Henry, J. (1996). \textit{Scientific Knowledge. A Sociological Analysis}. Athlone. London.

\bibitem[Bell(1987)]{Bel87a}
Bell, J. (1987). \textit{Speakable and Unspeakable in Quantum Mechanics}. Cambridge U. Press. Cambridge.

\bibitem[Bohm(1952)]{Boh52a}
Bohm, D. (1952). Phys. Rev. \textbf{85}, 166 and 180.

\bibitem[Bohr(1958)]{Boh58a}
Bohr, N. (1958).  \textit{Atomic Physics and Human Knowledge}. Wiley. New York.

\bibitem[Bohr(1963)]{Boh63a}
Bohr, N.(1963).  \textit{Essays 1958/1962 on Atomic Physics and Human Knowledge}. Wiley. New York.

\bibitem[Clauser \emph{et al.}(1969)]{ClaHorShiHol69a}
Clauser, J. F., Horne, M. A., Shimony, A., and Holt, R. A. (1969). Proposed Experiment to Test Local Hidden-Variable Theories. Phys. Rev. Letters \textbf{23}, 880-884.

\bibitem[Cohen-Tannoudji \emph{et al.}(1977)]{CohDiuLal77a}
Cohen-Tannoudji, C, Diu, B. and Laloë, F.(1977) \textit{M\'{e}canique Quantique}. Hermann. Paris.

\bibitem[Conway and Kochen(2006)]{ConKoc06a}
Conway J.H. and Kochen S., (2006). \textit{The Free Will Theorem.} Found. Phys. 36, 1441-1473.

\bibitem[Cramer(1986)]{Cra86a}
Cramer, J. G., (1986). \textit{The transactional interpretation of quantum mechanics}.  Rev. Mod. Phys. \textbf{58}, 647-688.

\bibitem[Dirac(1930)]{Dir30a}
Dirac, P. (1930). \textit{The Principles of Quantum Mechanics}. Third Edition. Clarendon. Oxford. p. vii.

\bibitem[Einstein \emph{et al.}(1935)]{EinPodRos35a}
Einstein, A., Podolsky, B. and Rosen, N.(1935).\textit{ Can Quantum mechanical Description of Physical Reality be Considered Complete?} Phys. Rev. \textbf{47}, 777.

\bibitem[Einstein(1970)]{Ein70a}
Einstein, A (1970). \textit{Albert Einstein: Philosopher-Scientist}. P. Schilpp Ed. Open Court.

\bibitem[Einstein(1971)]{Ein71a}
Einstein, A (1971). The \textit{Born-Einstein letters}. Macmillan, pp.\ 170-171.

\bibitem[Freire(2004)]{Fre04a}
Freire, O (2004).\textit{ The Historical Roots of ``Foundations of Quantum Physics'' as a Field of Research (1950-1970).} Found. Phys. \textbf{34}, 1741.

\bibitem[Gardner(1983)]{Gar83a}
Gardner, M (1983). \textit{The whys of a Philosophical Scrivener}. St. Martin's Griffin (2${}^{nd}$ edition 1999).

\bibitem[Ghirardi \emph{et al.}(1986)]{GhiRimWeb86a}
Ghirardi, G. C.,  Rimini, A. and Weber, T. (1986). Phys. Rev. \textbf{D34}, 470.

\bibitem[Heisenberg(1958)]{Hei58a}
Heisenberg, W. (1958).\textit{ Physics and Philosophy}. Harper and Row. New York.

\bibitem[Jammer(1974)]{Jam74a}
Jammer, M. (1974).\textit{The Philosophy of Quantum Mechanics}. John Wiley and Sons. New York.

\bibitem[Jaynes(2000)]{Jay00a}
Jaynes, J. (2000). \textit{The Origin of Consciousness in the Breakdown of the Bicameral Mind.} Mariner books. (Original from 1976).

\bibitem[Jennewein \emph{et al.}(2002)]{JenWeiPanZei02a}
Jennewein, T., Weihs, G., Pan, J., and Zeilinger, A. (2002). \textit{Experimental Nonlocality Proof of Quantum Teleportation and Entanglement Swapping}. Phy. Rev. Lett. \textbf{88}, 017903-3.

\bibitem[Lapiedra(2007)]{Lap07a}
Lapiedra, R. (2007). \textit{Las Carencias de la Realidad}. Tusquets, Ed. Barcelona.

\bibitem[Mermin(1983)]{Mer93a}
Mermin, D. (1993). \textit{Hidden variables and the two theorems of John Bell}. Rev. Mod. Phys., \textbf{65}, 803-812.

\bibitem[Peres(1978)]{Per78a}
Peres, A. (1978). \textit{Unperformed experiments have no results}. Amer. J. Phys. \textbf{46}. 745.

\bibitem[Rorty(1999)]{Ror99a}
Rorty, R. (1999). \textit{Philosophy and Social Hope}. Penguin Books.

\bibitem[Quine(1980)]{Qui80a}
Quine, W.V. (1980). \textit{From a Logical Point of View}, Logio-Philosophical Essays, 2${}^{nd}$ edition. Harvard University Press. Cambridge, MA.

\bibitem[Santos(2009)]{San09a}
Santos, E. (2009). \textit{Realist Interpretation of Quantum Mechanics}. arXiv:0912.4098v1.

\bibitem[Smolin(2007)]{Smo07a}
Smolin, L. (2007). \textit{The Trouble With Physics}. Mariner Books. Boston.

\bibitem[Snell(1982)]{Sne82a}
Snell, B. (1982). \textit{The Discovery of the Mind}. Dover. New York.

\bibitem[Squires(1990)]{Squ90a}
Squires, E. (1990).\textit{Conscious Mind in the Physical World}. Adam Hilger. Bristol. p. 117.

\bibitem[Stapp(1972)]{Sta72a}
Stapp, H. P. (1972).\textit{ The Copenhagen Interpretation}. Amer. J. Phys. \textbf{40}, 1098.

\bibitem[Weinberg(1993)]{Wei93a}
Weinberg, S. (1993).\textit{ Dreams of a final theory}. Vintage. London.

\bibitem[Weizs\"{a}cker(1971)]{Wei71a}
von Weizs\"{a}cker, C. F. (1971). In \textit{Quantum Theory and Beyond}. T. Bastin, Ed. Cambridge University Press. Page 26.

\bibitem[Werner(1989)]{Wer89a}
Werner, R. (1989). \textit{Quantum states with Einstein-Podolsky-Rosen correlations admitting a hidden-variable model}. Phys. Rev. A \textbf{40}, 4277.

\bibitem[Wheeler(1996)]{Whe96a}
Wheeler, J. A. (1996).\textit{ At Home in the Universe}. Springer. New York.

\bibitem[Whitrow(1950)]{Whi50a}
Whitrow, G. J. (1950). \textit{On the Foundations of Dynamics}. Br. J. Philos. Sci. \textbf{1} (2): 92-107.

\bibitem[Wigner(1967)]{Wig67a}
Wigner, E. (1967).\textit{ Symmetries and Reflections: Scientific Essays}. Indiana University Press. Bloomington.

\bibitem[Zukowski \emph{et al.}(1993)]{ZukZeiHorEke93a}
Zukowski, M., Zeilinger, A., Horne, M.A., and Ekert, A. K. (1993). ``\textit{Event-Ready Detectors'' Bell experiment with Entanglement Swapping. }Phys. Rev. Lett. \textbf{71}, 4287-90.

\end{thebibliography}
\end{document}